\documentclass[12pt,reqno]{amsart}

\usepackage{amsmath}
\usepackage{amssymb}
\usepackage{amsthm}
\usepackage{amscd, amsfonts, mathrsfs}

\usepackage{cases}
\usepackage{xfrac}

\usepackage[dvips]{epsfig}
\usepackage{verbatim} 

\usepackage{epsf}
\usepackage[bookmarksnumbered,pdfpagelabels=true,plainpages=false,colorlinks=true,
            linkcolor=black,citecolor=black,urlcolor=blue]{hyperref}

\theoremstyle{plain}
\newtheorem{theorem}{Theorem}[section]
\newtheorem{lemma}[theorem]{Lemma}
\newtheorem{proposition}[theorem]{Proposition}

\theoremstyle{definition}
\newtheorem{remark}[theorem]{Remark}
\newtheorem{notation}[theorem]{Notation}
\newtheorem{definition}[theorem]{Definition}

\newtheorem{assump}{Assumption}[section]

\numberwithin{equation}{section}

\newenvironment{ sys_eq }{\ left \ lbrace \ begin { array }{ @ {} l@ {}}}{\ end { array }\ right .}

\newcommand{\cD}{\mathcal D}

\newcommand{\cG}{\mathcal G}
\newcommand{\cH}{\mathcal H}

\newcommand{\cR}{\mathcal R}
\newcommand{\cS}{\mathcal S}
\newcommand{\cT}{\mathcal T}

\newcommand{\al}{\alpha}
\newcommand{\be}{\beta}
\newcommand{\ga}{\gamma}
\newcommand{\Ga}{\Gamma}
\newcommand{\de}{\delta}

\newcommand{\si}{\sigma}
\newcommand{\Si}{\Sigma}

\newcommand{\Om}{\Omega}

\newcommand{\RR}{\mathbb R}

\newcommand{\rar}{\rightarrow}

\newcommand{\vol}{\operatorname{vol}}
\newcommand{\id}{\operatorname{id}}
\newcommand{\dive}{\operatorname{div}}

\title[Quantization]{A note on  quantization in the presence of gravitational shock waves}

\author[Disconzi]{Marcelo M. Disconzi}
\address{Department of Mathematics\\
Vanderbilt University, Nashville, TN 37240, USA}
\email{marcelo.disconzi@vanderbilt.edu}

\begin{document}

\maketitle

\begin{abstract}
We study the quantization of a free scalar field when the background metric satisfies Einstein's
equations and develops gravitational shock waves.
\end{abstract}


\section{Introduction.}

The general notion of ``singularity," although not always precisely defined
and having different meanings in different situations, plays a crucial role in many 
physical theories. Broadly speaking, one expects that  
singularities will form in regimes where the system undergoes extreme 
dynamic conditions, e.g., turbulence in Fluid Dynamics or gravitational collapse in
General Relativity (GR).

One particular notion of singular behavior frequently encountered is 
that of a shock wave. Roughly, it 
corresponds to a discontinuity
in the solutions of the equations of motion --- in particular, it implies that
such solutions exist only in a generalized sense. In the case of Einstein's
theory of gravity, shock wave discontinuities happen in the first derivatives 
of the metric (see definition \ref{defi_GSW}). It follows that the spacetime
does not carry a ``smooth geometry,"  with the curvature tensor (which 
depends on second derivatives of the metric) being meaningfully defined
only in a distributional sense.

Although in a different context, 
physicists have been dealing with singular spacetimes
for quite a long time, with Hawking's singularity theorems \cite{HE} and its ramifications
(Cosmic Censorship, Penrose Inequality, etc.) being at the core of such
developments. It is widely believed that a successful quantum theory of gravity
will resolve many, if not all, of the difficulties and puzzles that arise in 
singular backgrounds. Unfortunately, despite much of the progress that has been 
witnessed in the last few decades, such a theory is not yet available.
This does not mean, however, that we cannot learn 
something about quantum effects in curved spacetimes, singular ones included. 
In the range where curvature effects 
 cannot be neglected but are still far from the Planck scale,
 the powerful (if yet difficult) machinery 
 of Quantum Field Theory (QFT) in curved backgrounds is available to us.
 To this day, some of the best hints of what the long sought quantum theory of gravity
might look like come from the study of quantum fields over a 
non-flat background that satisfies Einstein's equations. In fact, 
the ability of reproducing the black hole temperature --- discovered 
by Hawking via 
a careful application of  QFT on a background that undergoes classical
gravitational collapse \cite{H} --- is often regarded as the first test
for any theory attempting to quantize the gravitational field.
This, of course, is not different than many other instances in 
Physics where semi-classical formulations are useful in
providing insight into what the full quantum theory is, leading
in this way to fruitful directions of inquiry.

Furthermore, one should not have to wait until a full-fledged theory
of quantum gravity is in place in order to understand interesting 
physical phenomena,
which involve quantum effects
in a background where the singularities are still 
amenable to a fully classical treatment. This paper is a step in this direction.
We shall study the quantization of a free scalar field over a spacetime where gravitational
shock waves are formed. Einstein's equations will have to be defined in a suitable
weak sense, and some standard arguments, like the construction of propagators,
adapted to this weaker setting. A similar aspiration, namely, devise a quantization
scheme that can potentially incorporate 
singular background data, was explored in Ref. \cite{AD}.

The tools we shall employ are, in a sense, not new. Most of the constructions
go back to the work of Lichnerowicz on tensor distributions \cite{L1} 
(see Ref. \cite{D} for similar constructions and generalizations). These have 
the convenience of being simultaneously suited to the study of shock waves, on
one hand \cite{L2}, and to the quantization of fields in curved --- although smooth 
--- spacetimes, on the other hand \cite{L3}. Not surprisingly, many of the arguments
here presented consist of carefully checking that the 
results of Ref. \cite{L3} carry onto the framework of shock waves.
 As such  arguments can be done in charts, 
 our point of view will be purely local. Generalizations
to a  global setting are possible, provided that further conditions
are taken into account; we briefly comment on this at the end.

Although this paper focuses uniquely on mathematical aspects,
we stress that the study of gravitational shock waves in general, 
and corresponding scenarios where quantum effects might become important
in particular, has attracted significant attention in the Physics community 
(see e.g. Ref. \cite{DP1,DP2,DP3, HS, tH, HU, VV}
and references therein), hence the importance of laying out its mathematical foundations.

\begin{notation} (i) $\vol(g)$ denotes the volume form of the metric $g$,
and $|g|$ its determinant (in a local coordinate patch);
(ii) if $\Om \subseteq M$, $M$ a $C^\ell$ manifold, $\cD_{k,p}(\Om)$
denotes the space of $p$-tensors of class $C^k$, $0\leq k \leq \ell$,
 with compact support in $\Om$, i.e.,
test tensors in $\Om$; we shall use the given
 metric to identify covariant and contra-variant tensors,
hence referring  simply to ``$p$-tensors";
(iii) $\cD_{k,p}^\prime(\Om)$ is the space of continuous real-valued forms 
on $\cD_{k,p}(\Om)$,
where continuity is understood in the sense of the theory of distributions \cite{S};
(iv) we sometimes write  $\cD_p(\Om)$ 
when the differentiability is clear from the 
context, although for the most part it will suffice to deal with
elements in $\cD_{0,p}(\Om)$ and its dual;
(v) $\langle v, u \rangle$ means $u \in \cD_{k,p}(\Om)$ evaluated at 
$v \in \cD^\prime_{k,p}(\Om)$, whereas $\langle u_1, u_2 \rangle_g$ denotes 
the inner product between the $p$-tensors $u_1$ and $u_2$, although, when 
the distributions arise from locally integrable tensors, 
we naturally identify $\langle \cdot, \cdot \rangle$ 
and $\langle \cdot , \cdot \rangle_g$; 
(vi) we shall denote by $T^\prime$ the
tensor distribution defined by a locally integrable tensor $T$; 
(vii)$\nabla$ will denote
covariant differentiation with respect to the metric $g$. 
(viii) $[\al | \ga | \be ]$ indicates anti-symmetrization in 
the indices  $\al$ and $\be$.

Hypotheses 
in the form ``Assumption X'' (e.g., Assumption \ref{assumption_metric}) are
assumed throughout, therefore we do not state them in the theorems and definitions.
\end{notation}

\section{Gravitational shock waves.}
We start recalling some definitions and fixing out notation. Let $M$ be an
oriented 4-dimensional
differentiable  manifold of class $C^2$ and piecewise $C^4$, and $\Om \subset M$ 
a contractible open subset. 
Whenever coordinates are employed, it is implicitly assumed that $\Om$ is taken
small enough as to belong to the domain of a coordinate patch.
Let $\Si$ be a regular hypersurface in $\Om$ which is given locally as the level set
$\{ \varphi = 0 \}$ of a $C^1$ function $\varphi: \Om \rar \RR$; saying that $\Si$
is regular then means $d\varphi \neq 0$ on $\Si$. We suppose that $\Si$ partitions 
$\Om$ into two disjoint domains $\Om^+$ and $\Om^-$ given by
$\{ \varphi > 0 \}$ and $\{ \varphi < 0 \}$, respectively.

\begin{definition} Let $T$ be a $p$-tensor on $\Om$, and assume that $T$ 
is of class $C^0$ over $\Om^+$ and $\Om^-$. $T$ is \emph{regularly
discontinuous across}\footnote{\emph{``r\'eguli\`erment discontinues \`a la travers\'ee
de.} \cite{L2}"} $\Si$ if, when $\varphi \rar 0^+$
(i.e., $\varphi$ tends to zero through positive values) (resp. $\varphi \rar 0^-$),
$T$ converges uniformly to a $p$-tensor $T^+$ (resp. $T^-$) defined on $\Si$.
The \emph{discontinuity}\footnote{We avoid the more familiar notation $[ \cdot ]$ for the
discontinuity of a function in that it may cause confusion when written inside an equation.} $\{\{T \}\}$ of $T$ is the $p$-tensor on $\Si$ defined by
 $\{\{T\}\} = T^+ - T^-$.
\end{definition}

 Let $g$ be
Lorentzian metric\footnote{Our convention is $(+ \, - \, - \, -)$.}
of class $C^0$ and piecewise $C^2$ on $M$. We fix $\Om$, $\varphi$ and $\Si$ once
and for all, and henceforth assume for the rest of the paper:

\begin{assump} (i)  $\varphi$ is $C^3$ on $\Om^+$ and $\Om^-$, with second 
and third derivatives regularly discontinuous across $\Si$; (ii) $g$ is $C^0$
on $\Om$, $C^2$ on $\Om^+$ and $\Om^-$, and has first and second derivatives regularly discontinuous
across $\Si$.
\label{assumption_metric}
\end{assump} 

Unless stated otherwise, it will be assumed in this section 
that coordinates $\{ x^\al \}_{\al = 0}^3$
are adapted to $\Si$, what means that $x^0 = \varphi$ and $\partial_i$ is tangent
to $\Si$, $i=1,2,3.$ From now on, Greek indices run from $0$ to $3$, while Latin 
indices run from $1$ to $3$.

\begin{definition} In coordinates $\{ x^\al \}_{\al=0}^3$
 adapted to $\Si$, the metric components $g_{ij}$
will be called \emph{physical components with respect to $\Si$}, while $g_{0\al}$ will be 
called 
\emph{gauge}  (or \emph{non-physical}) \emph{components with respect to $\Si$}.
A change of coordinates $x^\prime{}^\al = x^\prime{}^\al(x^0,\dots,x^3)$
is called a \emph{change of gravitational gauge with respect to $\Si$}, or 
\emph{change of gauge} for short,
if it is the identity on $\Si$, and the the values of $\left.g_{\al\be}\right|_{\Si}$
 as well as of $\left. \partial_0 g_{ij}\right|_{\Si}$
are invariant under this change.
\end{definition}

\begin{remark} Despite the familiar terminology, the reader should not be led to think that
$\Si$ is a space-like hypersurface. In fact, for the case of interest in this paper, $\Si$
will be null-like.
\end{remark}

The necessity of treating $g_{ij}$ and $g_{0\al}$ differently comes from the 
well-known fact that GR has a gauge freedom due to the action 
of the diffeomorphism group of $M$. In fact, in a neighborhood of $\Si$
gauge changes take the form
\begin{gather}
x^\prime{}^\al(x^0,x^1,x^2,x^3) = x^\al + \frac{(x^0)^2}{2}\Big ( \partial_0 x^\prime{}^\al(0,x^1,x^2,x^3) + r^\al(x^0,x^1,x^2,x^3) \Big ),
\nonumber
\end{gather}
where $r^\al$ converges, along with its derivatives, uniformly to zero when $x^0 \rar 0$.
It is therefore possible to arrange the terms in parenthesis as to produce or eliminate
discontinuities of $\partial_0 g^\prime_{0\al}$ on $\Si$ --- showing that such 
terms carry no intrinsic physical meaning \cite{L2}. Furthermore, from Assumption \ref{assumption_metric}  one readily sees that
\begin{gather}
\{\{ \partial_i g_{\al\be} \}\} = 0 = \partial_i \{\{ g_{\al\be} \}\}.
\label{tang_der_g_zero}
\end{gather}
From this and the above discussion, we see that it suffices to focus on $\partial_0 g_{ij}$ in our study of discontinuities
of the metric across $\Si$ --- see definition \ref{defi_GSW}.

We now turn our attention to the appropriate notion of
weak solution for the Einstein's equations in the study of gravitational 
shock waves.  For each fixed pair of indices $\be\de$, the Christoffel symbols
$\Ga^\al_{\be\de}$ locally define  a vector field $\Ga^\al \equiv \Ga^\al_{\be\de}$,
which in turn, under
 our hypotheses, is locally integrable and hence defines a vector distribution $(\Ga^\al_{\be\de})^\prime$. 
 Following Lichnerowicz, it is therefore natural to define the \emph{curvature tensor distribution}
as
\begin{gather}
\cR^\al_{\be\ga\de} 
= \nabla_\ga (\Ga^\al_{\be\de})^\prime - \nabla_\de (\Ga^\al_{\be\ga})^\prime,
\nonumber
\end{gather}
where the covariant derivatives are interpreted in a distributional sense \cite{L1,L2}.
To find a simple formula for the associated distributional Ricci curvature and relate
it to the ordinary one, we shall use the following lemma, whose proof is an application
of the tools developed in Ref. \cite{L2}.

\begin{notation}
We put $\varphi_\ga = \partial_\ga \varphi$, thinking of these
as the components of the locally defined one form $d\varphi$.
\end{notation}

\begin{lemma} 
Let $T$ be a $p$-tensor of class $C^1$ on $\Om^+$ and $\Om^-$ such that 
$T$ and $\partial_\al T$ are regularly discontinuous across $\Si$. Then $\nabla T$  is 
regularly discontinuous across
$\Si$. Furthermore, $T$ and $\nabla T$ define tensor distributions $T^\prime$
 and $(\nabla T)^\prime$ such that
\begin{gather}
\nabla T^\prime - (\nabla T)^\prime =d\varphi \otimes \de_\Si \, \{\{ T \}\},
\nonumber
\end{gather}
where the covariant  derivative $\nabla T^\prime$ of $T^\prime$
is in the sense of distributions, 
and $\de_\Si$ is the Dirac delta on $\Om$ with support on $\Si$.
\label{lemma_cov_diff_T}
\end{lemma}
\begin{proof}
Let $A_\al^\be = \Ga_{\al\ga}^\be dx^\ga$ be the connection one form
of $g$ on $\Om$. From our hypotheses, we see that $A_\al^\be$ is $C^1$ in
$\Om^+$ and $\Om^-$; it is also regularly discontinuous across $\Si$.
It follows that $\nabla T$ is regularly discontinuous across $\Si$ since 
$\partial_\al T$ and $T$ are so.

Let $\chi_+$ (resp. $\chi_-$) be the function defined a.e. in $\Om$ which
equals to $1$ in $\Om^+$ (resp. $\Om^-$) and zero 
in $\Om^-$ (resp. $\Om^+$). Since $T$ and $\nabla T$ are in $L^1_{loc}(\Om)$,
they define tensor distributions,  which can be written as
\begin{gather}
T^\prime = \chi_+ T + \chi_- T,
\label{T_prime_lemma}
\end{gather}
and 
\begin{gather}
(\nabla T)^\prime = \chi^+ \nabla T + \chi_- \nabla T,
\label{nabla_T_prime_lemma}
\end{gather}
where these equalities are to be understood in the sense of 
distributions\footnote{Equalities among quantities in $\cD^\prime_p(\Om)$ are by definition
in the distributional sense, so we shall no longer write ``in the sense of distributions."}. 
We can write
$\vol(g) = d\varphi \wedge \omega$. Notice that $\omega$ depends on $\varphi$, but
if $\eta$ is another $3$-form such that $\vol(g) = d\varphi \wedge \eta$, then
$\eta = \omega + d\varphi \wedge \si$ for some $2$-form $\si$. Hence, for
any test function $u$ 
\begin{gather}
- \int_{\partial \Om^+} u \, \omega = \int_{\partial \Om^-} u \, \omega
\nonumber
\end{gather}
\noindent (where $\partial \Om^+$ and $\partial \Om^-$ are the oriented boundaries 
so that $\partial \Om^+ = - \partial \Om^-$)
has a well-defined value independent of the choice of $\omega$. In particular,
$\langle \de_\Si, u\rangle$ is given by
\begin{gather}
\langle \de_\Si, u\rangle = - \int_{\partial \Om^+} u \, \omega
 = \int_{\partial \Om^-} u \, \omega \, .
\nonumber
\end{gather}
Notice that because $\omega$ is continuous this gives that $\de_\Si$ is 
in fact an element of $\cD_{0,0}^\prime(\Om)$. Since $\de_\Si$ can be multiplied by
locally integrable functions, this same formula also shows that 
$\de_\Si$ has a well-defined action on $\cD_{k,0}(\Om)$, $k=1,2$.

Clearly $\nabla \chi_+ \in \cD^\prime_{1}(\Om)$, and for any $u \in \cD_1(\Om)$
\begin{align}
\begin{split}
\langle \nabla \chi_+, u \rangle &= 
\langle \nabla \chi_+, u \rangle_g =
-\langle \chi_+, \dive u \rangle_g =
- \int_{\Om^+} \frac{1}{\sqrt{-|g|}} \partial_\al (\sqrt{-|g|} u^\al ) \vol(g) \\
& = - \int_{\Om^+}  \partial_\al (\sqrt{-|g|} u^\al ) dx^0 \wedge dx^1 \wedge dx^2 \wedge
dx^3 \\
& =- \int_{\partial \Om^+}  \sqrt{-|g|} u^0 dx^1 \wedge dx^2 \wedge dx^3 \\
& = - \int_{\partial \Om^+}  \varphi_\al u^\al \, \omega = \langle d\varphi \,
\de_\Si, u \rangle,
\end{split}
\nonumber
\end{align}
so that $\nabla \chi_+ = d\varphi \, \de_\Si$. Similarly, one finds
 $\nabla \chi_- = - d\varphi \, \de_\Si$. Therefore,
 \begin{gather}
 \nabla (\chi^+ T) = (\nabla \chi_+ ) \otimes T + \chi_+ \nabla T =
  d\varphi \, \de_\Si \otimes T + \chi_+ \nabla T.
  \nonumber
 \end{gather}
Computing a similar expression for $ \nabla (\chi^+ T)$, using 
(\ref{T_prime_lemma}) and (\ref{nabla_T_prime_lemma}), yields the result.

\end{proof}
The discontinuities of $g$ will enter in the expression for 
the distributional curvature. The following lemma is useful to handle them.

\begin{lemma}
For each of the functions $g_{\al\be}$, there exists a 
scalar distribution $d_{\al\be}$ such that 
\begin{gather}
\de_\Si \{\{ \partial_\ga g_{\al\be} \}\} = \varphi_\ga d_{\al\be}. 
\nonumber
\end{gather}
\label{lemma_d_b}
\end{lemma}
\begin{proof} 
Arguing similarly to the proof of Lemma \ref{lemma_cov_diff_T} we see that
\begin{align}
\langle \nabla_i \delta_\Si, u \rangle & =  \int_{\partial \Om^+} \partial_i (\sqrt{-|g|}\, u)
\, dx^1\wedge dx^2 \wedge dx^3 
\nonumber \\
& =  \int_{\Om^+} \partial_i \partial_0 (\sqrt{-|g|}\, u) \,
dx^0 \wedge  dx^1\wedge dx^2 \wedge dx^3  \nonumber 
\\
 & =  0,
\nonumber
\end{align}
where we used that $u$, being a test function, is compactly supported
in $\Om$. A similar result holds in $\Om^-$ and we conclude that
 $\nabla_i \de_\Si = 0$ (in a more invariant fashion, we 
 can say $\nabla_X \de_\Si = 0$
 for vector fields $X$ tangent to $\varphi=$constant). It follows that
 \begin{gather}
 \nabla \de_\Si = d\varphi \nabla_0 \de_\Si.
 \nonumber
 \end{gather}
 This combined with (\ref{tang_der_g_zero}) (and, of course,
 our assumptions on $g$) produces the desired result.
\end{proof}
Intuitively, $d_{\al\be}$ corresponds to the ``jumps" that come when 
we try to differentiate $\de_\Si \{ \{ g_{\al\be} \} \}$ across $\Si$.
 It is not difficult to see that the existence of $d_{\al\be}$ is equivalent to 
the formulation given in Ref. \cite{L2}, namely, to find, for each $g_{\al\be}$,
a function $b_{\al\be}$ defined over $\Si$ such that
\begin{gather}
\{\{ \partial_\ga g_{\al\be} \}\} = \varphi_\ga b_{\al\be} \text{ on } \Si 
\label{def_b}
\end{gather}
and satisfying\footnote{Notice that, in principle, $b_{\al\be}$ is defined only on $\Si$,
but 
this suffices since $\de_\Si$ is supported on $\Si$; alternatively
$b_{\al\be}$ can be extended to $\Om$ without affecting the distributional
equality $d_{\al\be} = \de_\Si b_{\al\be}$.}  $d_{\al\be} = \de_\Si b_{\al\be}$.
\begin{remark}
$b_{\al\be}$ depends on the choice of equation for $\Si$. The functions
$b_{\al\be}$ were initially introduced by Lichnerowicz \cite{L2} and are important 
in determining the conditions of shock for gravitational waves. Such conditions will
appear below when we construct the propagator for the Klein-Gordon equation.
\end{remark}

From Lemma \ref{lemma_cov_diff_T}, Lemma \ref{lemma_d_b} and the above considerations
 it follows that 
\begin{gather}
\cR_{\al\be\ga\de} = R_{\al\be\ga\de}^\prime + \de_\Si \cH_{\al\be\ga\de},
\nonumber
\end{gather}
where the distribution $R_{\al\be\ga\de}^\prime$, with 
$R_{\al\be\ga\de}$ being the curvature tensor, is well-defined 
 because of Assumption \ref{assumption_metric},
and $\cH_{\al\be\ga\de}$ is the tensor on $\Si$ given by
\begin{gather}
-2 \cH_{\al\be\ga\de} = b_{\al\ga} \varphi_\be \varphi_\de  
- b_{\al\de}\varphi_\be \varphi_\ga + b_{\be\de}\varphi_\al\varphi_\ga -
b_{\be\ga} \varphi_\al \varphi_\de.
\nonumber
\end{gather}
The distributional Ricci tensor is defined as
\begin{gather}
\cR_{\al\be} = R_{\al\be}^\prime +  \de_\Si \cH_{\al\be},
\label{dist_Ricci}
\end{gather}
where
\begin{gather}
2 \cH_{\al\be} = b_{\al\mu} \varphi^\mu \varphi_\be
+  b_{\be\mu }\varphi^\mu \varphi_\al - b^\mu_\mu \varphi_\al \varphi_\be 
- b_{\al\be} \varphi^\mu \varphi_\mu,
\label{H_al_be}
\end{gather}
so that $\cR_{\al\be}$ is formally the trace of the distributional curvature tensor.

Recall that given a stress-energy tensor $\cT_{\al\be}$ 
(possibly identically zero), one can write Einstein's equations as
\begin{gather}
R_{\al\be} = \kappa \rho_{\al\be},
\label{Einstein_eq}
\end{gather}
where $\kappa$ is a constant and 
\begin{gather}
\rho_{\al\be} = \cT_{\al\be} - \frac{1}{2} \cT g_{\al\be},
\nonumber
\end{gather}
with $\cT$ the trace of $\cT_{\al\be}$. The last ingredient we need to define
the distributional Einstein's equations is the regularity of $\cT_{\al\be}$.

\begin{assump} $\cT_{\al\be}$ is a given symmetric two-tensor, continuous 
on $\Om^+$ and $\Om^+$ and regularly discontinuous across $\Si$.
\end{assump} 
It follows that $\rho_{\al\be}$ shares the same regularity
properties of $\cT_{\al\be}$. 
The reason why we think of $\cT_{\al \be}$ as given is that
existence of solutions to the distributional Einstein's equations ---
in which case 
$T_{\al\be}$ has a determined functional form but depends on 
the metric and the matter fields of the problem --- will not be investigated;
rather we shall assume we are given a solution that defines a background
on which fields will be quantized. 

\begin{definition} The distributional Einstein's equations are defined as
\begin{gather}
\cR_{\al\be} = \kappa \rho_{\al\be}^\prime.
\label{dist_Einstein}
\end{gather}
\end{definition}

\begin{remark} When the metric and $\cT_{\al\be}$ are sufficiently regular (say, $C^2$) on the
whole of $\Om$, we see, from 
(\ref{def_b}), that the functions $b_{\al\be}$ vanish identically. Also,
$R_{\al\be}^\prime$ and  $\rho_{\al\be}^\prime$ 
can be identified with the classical 
$R_{\al\be}$ and  $\rho_{\al\be}$. From (\ref{dist_Ricci}) and 
(\ref{H_al_be}) it then follows that (\ref{Einstein_eq}) and
(\ref{dist_Einstein}) agree.
\end{remark}

\begin{definition} The hypersurface $\Si$ is called a \emph{wave front} and is said
to define a \emph{gravitational shock wave} if $g$ satisfies 
(\ref{dist_Einstein}), and the first derivatives of the physical components 
of $g$ are discontinuous across $\Si$, and also
regularly discontinuous across $\Si$.
\label{defi_GSW}
\end{definition}
It is possible to show that a shock wave $\Si$ is necessarily 
null-like\footnote{Recall that hydrodynamic shock waves have the property of being 
supersonic before the shock and subsonic after the shock, with ordinary waves propagating
at the sound speed. Here, the speed of light plays the role of the sound speed, hence
shocks cannot be ``superluminal" before the shock. 
};
had it not been null-like, and hence characteristic to the reduced Einstein equations, 
the values of the induced metric and its derivatives
would uniquely determine a regular solution on both sides of $\Si$, preventing 
discontinuities. See Ref. \cite{L2} for details.

\section{Quantization.}
In this section, we adapt the techniques of Ref. \cite{L3}
to the weak setting developed above. We assume the same hypotheses
and notation as before. We shall deal only with a 
free scalar field, although it is very likely that such results can be generalized to 
tensor and spinor fields. Not surprisingly, the extension to interacting theories, 
on the other hand,
is expected to pose severe difficulties. For the rest of the paper, we let $m > 0$
be a fixed parameter. For a compact set $K \subset \Om$, we denote by 
$C^+(K)$ (resp. $C^-(K)$) the future (resp. past) of $K$ in the usual sense of 
GR, and the cone of $K$ the set
$C(K) = C^+(K)\cup C^-(K)$.
 As our point of view is purely local, strictly speaking $C^+(K)$, is
the future of $K$ within $\Om$ (analogously for $C^-(K)$).

\begin{proposition} For each fixed $x \in \Om$, there exist two elementary kernels
$E^\pm_x$ satisfying 
\begin{gather}
(\Box_g + m^2) E^\pm_x = \de_x,
\nonumber
\end{gather}
where $\de_x$ is the Dirac delta supported at $x$. $E^+_x$ (resp. $E^-_x$) is 
unique and has support on $C^+(x)$ (resp. $C^-(x)$).
\label{prop_elem_kernels}
\end{proposition}
\begin{remark} In Minkowski space, $E^\pm$ are the standard 
advanced and retarded kernels.
\end{remark}
\begin{proof} 
We give the proof for $E^+$, with the existence of $E^-$ being 
completely analogous. First, notice that the standard formula
\begin{gather}
\int_{\Om} u \Box_g v  \vol(g)
= \int_\Om u \partial_\mu (g^{\mu\nu} \sqrt{-|g|} \partial_\nu v ) \, 
dx^0\wedge dx^1\wedge dx^2 \wedge dx^3 =
\int_\Om v \Box_g u  \vol(g)
\nonumber
\end{gather}
still holds for all $u,v \in \cD_{2,0}(\Om)$,  
despite $g$ being only continuous on the whole of $\Om$. 
This can be verified by performing integration by parts on $\Om^+$ and $\Om^-$
separately. The continuity of the metric across $\Si$ guarantees that the
resulting boundary integrals (over $\partial \Om^+$ and $\partial \Om^-$)
 cancel out.
$\Box_g$ is,
therefore, formally self-adjoint, hence we seek to find 
$E^+_x$ such that for all $u \in \cD_{2, 0}(\Om)$
\begin{gather}
\langle E^+_x, (\Box_g  + m^2) u\rangle = \langle \de_x, u \rangle.
\nonumber
\end{gather}
Following Choquet-Bruhat \cite{B}, we shall construct a parametrix  $\si^+_x$, 
i.e., a distribution satisfying\footnote{The terminology  ``parametrix"  was introduced by Hilbert and usually indicates an approximation for the fundamental solution of a linear equation.}
  \begin{gather}
 (\Box_g + m^2) \si_x^+ = \de_x - Q^+_x,
 \label{parametrix}
 \end{gather}
 where $Q^+_x$ is an integrable function supported on $\Ga_x \equiv \partial C^+(x)$.
 Indeed, we shall show that the fixed point argument employed in Ref.
 \cite{B}, where it is 
 assumed that the metric is sufficiently differentiable, still goes 
 through under our assumptions.

We start by investigating the restrictions that (\ref{dist_Einstein}) imposes
on the the quantities $b_{\al\be}$ --- these are the so-called conditions of shock, 
originally studied in Ref.  \cite{L2}. Because $g$ is $C^2$ over $\Om^+$ and $\Om^-$,
and in light of Assumption \ref{assumption_metric}, we have that
(\ref{Einstein_eq}) is satisfied in $\Om^+\cup\Om^-$ and, moreover,
\begin{gather}
R_{\al\be}^\prime = \kappa \rho_{\al\be}^\prime \, \text{ in } \, \Om.
\nonumber
\end{gather}
From this, (\ref{dist_Ricci}), (\ref{H_al_be}) and (\ref{dist_Einstein}) we conclude that
\begin{gather}
\varphi^\mu b_{ [ \al | \mu | } \varphi_{\be ] } = 0 \, \text{ on } \, \Si.
\label{conditions_shock}
\end{gather}
We now claim that $\varphi^\mu \nabla_\mu$ is a well-defined operator along $\Si$. 
In fact,  direct computation gives that  on $\Si$
\begin{gather}
\varphi^\mu \{ \{ \Ga_{\al \mu}^\be \} \} 
= \varphi^\mu b_{ [ \al | \mu | } \varphi_{\be ] }, 
\nonumber
\end{gather}
which vanishes by 
(\ref{conditions_shock}). From this, the regularity of $\varphi$ in 
Assumption \ref{assumption_metric} and the fact that $\Si$ is 
null-like \cite{L2}, it follows not only that $\varphi^\mu \nabla_\mu$ is
 well-defined 
over $\Si$ but also that
\begin{gather}
\varphi^\mu \nabla_\mu \varphi^\al = 0 \, \text{ on } \, \Si.
\nonumber
\end{gather}
As a consequence, geodesics are well-defined over and  span $\Si$.
It follows that $\Ga_x$ coincides with $\Si$, in the neighborhood of $x$, 
if $x \in \Si$ (recall that $\Si$ is null-like). Due to  discontinuities 
on derivatives of $g$ across $\Si$, this is the relevant situation
 where we need
to show that the proof of Choquet-Bruhat in Ref. \cite{B} can be adapted.
No difficulty arises when $x \notin \Si$, 
henceforth we assume that $x$ belongs to $\Si$.

We claim that having $\Ga_x = \Si$ is the crucial point that enables
us to mirror  Choquet-Bruhat's construction \cite{B} of 
$\si^+_x$. In fact, as it is briefly reviewed below, the parametrix
is first obtained over $\Ga_x$, and then extended. 
This naturally involves derivatives along
$\Ga_x$ of the coefficients of the differential operator under consideration.
In our case these coefficients are the metric components $g_{\al\be}$. Because $\Ga_x = \Si$ near 
$x$, such derivatives produce, under our assumptions, well defined and sufficiently regular functions 
(see (\ref{tang_der_g_zero}) and prior discussion).
Bearing this in mind, the proof follows very closely that of 
Ref. \cite{B}, thus we shall give only its main steps, 
stressing the key points where the low regularity
of $g$ is circumvented.

The parametrix $\si^+_x$ is given as the extension of a distribution  $\widetilde{\si}_x^+$
supported on $\Ga_x$, i.e., for any test function $u$, 
\begin{gather}
\langle \si_x^+, u \rangle = \langle \widetilde{\si}_x^+, 
\left. u \right|_{\Ga_x} \rangle_{\Ga_x},
\label{parametrix_tilde}
\end{gather}
where  $\langle \cdot , \cdot   \rangle_{\Ga_x}$ denotes 
the pairing for test functions defined on $\Ga_x$. $\widetilde{\si}_x^+$,
in turn, is obtained as (an approximation for) a fundamental solution
for a Laplace type operator $\widetilde{\Delta}$ over $\Ga_x$, which 
involves only the components $g_{ij}$ of the metric (the precise form
of $\widetilde{\Delta}$ is given in Ref. \cite{B-Long}). Experience suggests
the Ansatz
\begin{gather}
\widetilde{\si}_x^+ = \frac{ B_x }{ \tau },
\label{Ansatz}
\end{gather}
where $B_x$ is a sufficiently regular function to be 
determined, and $\tau$ is the 
canonical parameter\footnote{See e.g. Ref. \cite{B-Book}; essentially,
$\tau$ is a parameter
obtained when null geodesics
are viewed as projections on spacetime of bicharacteristics 
associated with the eikonal equation on the cotangent space.}
for the null geodesics, issued from $x$, that span $\Ga_x$;
notice that $\widetilde{\si}_x^+$ is singular at $x$ (where $\tau=0$).

With the Ansatz (\ref{Ansatz}) at hand, we compute $\widetilde{\Delta} \widetilde{\si}_x^+$
and work the above steps backwards, until we reach (\ref{parametrix}).
The validity of (\ref{parametrix}) requires certain derivatives in the 
direction transverse to $\Ga_x$ to drop out. Imposing this leads to the vanishing
of the coefficients  of such derivatives; we symbolically
denote these terms by $\mathcal{C}$ (again, see Ref. \cite{B-Long} for 
the precise expression). These coefficients involve $B_x$, the metric
and derivatives of both. Imposing $\mathcal{C} = 0$ determines a first order 
differential equation for $B_x$, which can be
rewritten as an integral equation. The fact that  derivatives of the metric
occur solely along $\Ga_x$ gives, in light of (\ref{tang_der_g_zero}),
 enough regularity to apply standard techniques
and solve this integral equation\footnote{We remark that many of the aforementioned steps are not immediately 
apparent on Ref. \cite{B}, but the reader can consult 
Ref. \cite{B-Long} where fairly detailed calculations are provided.}. Finding such a solution determines 
$Q_x^+$. 

From equality (\ref{parametrix}), it now follows that
\begin{gather}
u(x) = \int_{\Om} Q_x^+ u \vol(g) + \langle \si_x^+ , (\Box_g + m^2) u\rangle.
\nonumber
\end{gather}
This can be viewed as an integral equation for $u$, which can be solved by a 
Neumann-series type of technique as in chapter V of Ref. \cite{dR}, possibly after
shrinking $\Om$.  
The solution is unique and operates linearly and continuously over $u$, yielding
the desired kernel.
\end{proof}

Besides the presence, in $\mathcal{C}$,  of derivatives only along $\Ga_x$, 
another feature necessary for Choquet-Bruhat's
proof to carry over our setting is the dimensionality of the spacetime.
In $n$ spacetime dimensions\footnote{For $n$ even. The case $n$ odd is obtained by reduction from
an even dimensional space.}, (\ref{parametrix_tilde}) is replaced by
\begin{gather}
\langle \si_x^+, u \rangle = 
\sum_{j=0}^{\frac{n}{2} - 2 } \langle \widetilde{\si}_{j,x}^+, 
\left. \partial_n^j u \right|_{\Ga_x} \rangle_{\Ga_x},
\nonumber
\end{gather}
where each distribution $\widetilde{\si}_{j,x}^+$ takes a form analogous 
to (\ref{Ansatz}), with the power of $\tau$ now depending 
on $n$ and $j$. Following similar arguments shows that the integral equation
to be solved now involves the term
\begin{gather}
\sum_{j=0}^{\frac{n}{2} - 2 }  \langle \si_x^+ ,  \partial_n^j (\Box_g + m^2) u\rangle.
\nonumber
\end{gather}
This will lead to an ill-defined expression unless
further hypothesis on $g$ are considered.

\begin{remark}
The core part of the argument for the existence of the elementary
kernels $E^\pm_x$ can be traced back to some compatibility conditions
for the metric induced on $\Si$ (essentially (\ref{tang_der_g_zero}),
and the continuity of $g$ across $\Si$). If one adopts the point of view that $\Om^\pm$
are two separated spacetimes with a null boundary, these compatibility conditions 
indicate when it is possible to glue $\Om^+$ and $\Om^-$ along their boundaries, 
in a way that the pre-existing structures on $\Om^\pm$ extend
to $\Om = \Om^+ \cup \Si \cup \Om^-$.
This is exactly the approach taken by Clarke and Dray in Ref. \cite{CD},
where  conditions for such a gluing are studied. Although their
setting is significantly more general than ours --- their manifolds are
$C^1$ (piecewise $C^3$), while here we employ $C^2$ (piecewise $C^4$),
and they do not impose the field equations,  consequently 
the notion of a shock wave does not
play any role ---, surprisingly the only
necessary and sufficient condition for the two spacetimes to be joined is 
that the naturally induced three-metrics on the boundaries agree. We notice that,
 although similar 
results had been  known for space-like boundaries (see references 
in Ref. \cite{CD}), Clarke and Dray's result applies to all types of hypersurfaces.
It would be interesting to investigate if Proposition \ref{prop_elem_kernels} 
can be generalized to a setting similar to that of Ref. \cite{CD}. 
This would likely require a 
considerably different proof than the one presented here in
that the shock wave structure, which implies (\ref{tang_der_g_zero}),
has been used. Recent works on the the geometry of null 
hypersurfaces (e.g. Ref. \cite{HD} and references therein)
are likely to be important in this regard  (see also Ref. \cite{Cl},
where the notion of global hyperbolicity on low regularity spacetimes
is discussed in connection with the solvability of the wave equation
with rough data).
\end{remark}

As usual, we think of the elementary kernels $E^\pm_x$ as defining a distribution  
$E^\pm(x^\prime, x)$ in
$\Om \times \Om$ by
\begin{gather}
\langle E^\pm(x^\prime, x), u(x^\prime, x) \rangle_{\Om \times \Om} 
= \langle E^\pm_{x^\prime}(x), u(x^\prime, x) \rangle_{\Om \times \Om} 
.
\end{gather}
The lemma below verifies that the standard symmetry properties of these distributions
known to be true in the smooth setting,
continue to hold under our hypothesis.

\begin{lemma} Let $E^\pm(x,x^\prime)$ be as above. Then
\begin{gather}
E^+(x,x^\prime) = E^-(x^\prime,x),
\nonumber
\end{gather}
for $x,x^\prime \in \Om$.
\label{lemma_symm_kernel}
\end{lemma}
\begin{proof} Fix a test function $u$ and define
\begin{gather}
v_1(x^\prime) = \langle E^-(x, x^\prime), u(x) \rangle.
\nonumber
\end{gather}
Since $E^-_{x}(x^\prime)$ has support on the past of $x$, $v_1$ is supported
on the past of the support of $u$. We have
\begin{gather}
(\Box_{x^\prime} + m^2) v_1(x^\prime) =
 \langle(\Box_{x^\prime} + m^2) E^-(x, x^\prime), u(x) \rangle
 = \langle \de(x,x^\prime), u(x) \rangle = u(x^\prime),
\nonumber
\end{gather}
where $\Box \equiv \Box_g$, $\Box_x$ means that derivatives are with respect
to the $x$ variable, and we used the symmetry of the Dirac delta.

For sufficiently smooth metrics, Choquet-Bruhat proved \cite{B_comm_Leray}
that any solution $v_2$ of 
\begin{gather}
(\Box + m^2)v_2 = u
\label{eq_v_2}
\end{gather}
with support compact towards the future\footnote{$K \subset \Om$ is said compact
towards the future if $C^-(K) \cap C^+(x)$ is compact of empty for each $x \in \Om$;
a similar notion applies for compact towards the past.} is given by
\begin{gather}
v_2(x^\prime) = \langle E^+(x^\prime,x), u(x) \rangle,
\label{rep_sol}
\end{gather}
provided that $\Om$ is taken sufficiently small.
An inspection on her proof shows that the same statement still holds in our setting.
In fact, the set $K = C_-(\operatorname{supp}(u)) \cap C_+(x)$ is compact (or empty).
Choose a function $z \in \cD_{2,0}(\Om)$ equal to $1$ in a compact neighborhood
of $K$. If $v_2$ satisfies (\ref{eq_v_2}), we obtain that for any 
$x^\prime \in C_-(\operatorname{supp}(u))$,
\begin{gather}
\langle E^+(x^\prime,x), u(x) \rangle = 
\langle E^+(x^\prime,x),(\Box _x+ m^2)v_2(x)\rangle 
\nonumber 
\\
=
\langle E^+(x^\prime,x),(\Box _x+ m^2)( z(x) v_2(x))\rangle.
\nonumber
\end{gather}
But, on the other hand, if
\begin{gather}
\widetilde{v}_2(x^\prime) = \langle E^+(x^\prime,x), u(x) \rangle,
\nonumber
\end{gather}
then
\begin{gather}
\langle E^+(x^\prime,x), u(x) \rangle = 
\langle  (\Box _x+ m^2) E^+(x^\prime,x), z(x) v_2(x) \rangle = \widetilde{v}_2(x^\prime),
\nonumber
\end{gather}
which shows (\ref{rep_sol}). We conclude that $v_1 = v_2$, from what the result follows.
\end{proof}
Following standard convention and terminology, we then define
the \emph{propagator} of $\Box_g + m^2$ as the distribution on $\Om \times \Om$
given by
\begin{gather}
G(x,x^\prime) = E^+(x^\prime,x) - E^-(x^\prime,x).
\nonumber
\end{gather}
It is seen that
\begin{gather}
(\Box_x + m^2 ) G(x,x^\prime) =0.
\nonumber
\end{gather}
In light of Lemma \ref{lemma_symm_kernel}, 
\begin{gather}
G(x, x^\prime) = - G(x^\prime,x).
\nonumber
\end{gather}

Consider the massive wave equation in $\Om$,
\begin{gather}
(\Box_g + m^2)u = 0.
\label{wave}
\end{gather} 
Let us denote by $D^+(K)$ (resp. $D^-(K)$) the future (resp. past) domain of 
dependence of an achronal set $K$, and $D(K) = D^+(K) \cup D^-(K)$.
We next suppose that:

\begin{assump} There exists in $\Om$ a kernel $\cG(x,x^\prime)$, which is a 
solution (in $x$) of (\ref{wave}), satisfying $\cG(x,x^\prime) = \cG(x^\prime,x)$,
and such that, for each space-like three surface  $\cS \subset \Om$
and $x,x^\prime \in D(\cS) \cap \Om$,
the following holds:
\begin{gather}
G(x,x^\prime) = \int_{\cS} \Big( \cG(x,y) \partial_\mu \cG(x^\prime,y) - \cG(x^\prime,y) 
\partial_\mu \cG(x,y) \Big)\, dA_g^\mu(y),
\label{main_formula}
\end{gather}
where $dA_g$ is the volume element induced on $\cS$ by $g$.
\label{main_assump}
\end{assump}
\begin{remark} In Minkowski space, $\cG$ is the so-called $D^1$ distribution
associated with the Pauli-Jordan propagator \cite{BS}.
\label{Pauli-Jordan}
\end{remark}
As in the usual case of smooth coefficients, for $x \in D(\cS) \cap \Om$
we have the following formula for a solution $u$ of the Cauchy problem of
(\ref{wave}) \cite{B,B_comm_Leray},
\begin{gather}
u(x) = \int_{\cS} 
\Big( u(y) \partial_\mu G(x,y) - G(x,y) \partial_\mu u(y) \Big)\, dA_g^\mu(y).
\nonumber
\end{gather}
This determines the solution $u$ in terms of the Cauchy data
$q:= \left. u \right|_\Si$ and $p:= \left. \partial_\nu u \right|_\Si$.
The pair $(q,p)$ plays the role of a point on phase space. 
We now define $\widetilde{u}$, depending on $u$, by
\begin{gather}
\widetilde{u}(x) = \int_{\cS} 
\Big( u(y) \partial_\mu \cG(x,y) - \cG(x,y) \partial_\mu u(y) \Big)\, dA_g^\mu(y),
\nonumber
\end{gather}
and introduce the norm
\begin{gather}
(u,u) = \int_{\cS} 
\Big( u(y) \partial_\mu \widetilde{u}(y)  - \widetilde{u}(y) \partial_\mu u(y) \Big)\, dA_g^\mu(y).
\nonumber
\end{gather}
Define
\begin{gather}
G^+ = \frac{G - i \cG}{2}, \,\,\, G^- = \frac{ G + i \cG}{2},
\nonumber
\end{gather}
and
\begin{gather}
u^+ = \frac{u - i \widetilde{u}}{2}, \,\,\, u^- = \frac{ u + i \widetilde{u}}{2}.
\nonumber
\end{gather}
\begin{proposition} The following identity holds
\begin{gather}
(u^+,u^+) = (u^-,u^-) = \frac{1}{2} (u,u).
\nonumber
\end{gather}
\label{prop_iden}
\end{proposition}
\begin{proof} This follows directly from the above formulas.
\end{proof}
We can now carry out the
quantization of a free scalar field over the background $(\Om,g)$ by
following the corresponding steps in the quantization of fields defined over a smooth
curved background, as originally proposed by Lichnerowicz \cite{L3}, 
and extended by other authors, in particular Wald \cite{W} (see also Ref. \cite{HW}). 
In order to obtain a well-defined quantization procedure, one has to
be  specific about the structure of the operators involved, indicating
their domains, self-adjointness properties etc. However, 
with the propagator $G(x,x^\prime)$, Proposition 
\ref{prop_iden}, and formula (\ref{main_formula}) at hand, 
these constructions are the same as in Ref. \cite{L3} (see also 
Ref. \cite{W}). Hence, for the sake of simplicity, 
we shall restrict ourselves to indicating what the main features of the  corresponding 
quantum theory are.
But before doing that we
first recall, rather briefly, some core features of the canonical 
quantization of a scalar field in Minkowski space (as can found, e.g.,
in Ref. \cite{BS}), so that the reader 
less familiar with Ref. \cite{L3} will be able to see the close parallel.

In Minkowski space, solutions $\phi$ to 
$(\Box + m^2)\phi = 0$ can be written, with the
help of Fourier transform and using standard notation, as
\begin{gather}
\phi(x) = \phi^+(x) + \phi^-(x) = \int e^{ik \cdot x} \phi^+(\vec{k}) c(k^0) \,  d \vec{k}  +
\int e^{-ik \cdot x} \phi^-(\vec{k})c(k^0) \, d \vec{k},
\nonumber
\end{gather}
where $\cdot$ is the Lorentz inner product, $k^0 = \sqrt{\vec{k}^2 + m^2}$, $c(k^0)$ is 
a normalization factor, and the functions $\phi^\pm(\vec{k})$ are constructed out of the Fourier
transform of $\phi$ (see Ref. \cite{BS}). A map $\phi \mapsto \phi^\pm$ is naturally
obtained in this way, with the fields $\phi^\pm$ called the positive and negative
energy components of the field $\phi$; these, in momentum representation and upon
quantizing, are associated with the creation and annihilation operators of the theory. Furthermore,
the corresponding quantum fields (operators) obey canonical commutation relations that
are postulated in a prescribed fashion out of the Poisson brackets of the classical theory.
For instance, one has
\begin{gather}
[ \phi^-(x), \phi^+(x^\prime) ] = -i D^-(x - x^\prime),
 \nonumber
\end{gather}
where $D^-$ is the negative energy component of the Pauli-Jordan propagator
(see remark \ref{Pauli-Jordan}). The reader can consult the standard literature
(e.g. Ref. \cite{BS}) to refresh his or her memory of the canonical quantization of 
a scalar field in Minkowski space.

Keeping the quantization in Minkowski space in mind for the purpose of 
analogy, we turn attention back to the formulation treated in this paper. We have:
\vskip 0.25cm
\noindent \textbf{I.} 
The association
$u \mapsto \widetilde{u}$ defines a linear automorphism $J$ on the space of 
solutions\footnote{A $2$-form $\omega$ can also be introduced  
in the space of solutions, and shown to  define a symplectic structure. $\omega$ and
$J$ are compatible and define an (infinite dimensional)  K\"ahler structure. These 
intrinsic
quantities can be used to  derive an appropriate
notion of inner product to carry out the quantization procedure. See 
Ref. \cite{L3} and \cite{W} for details.} of (\ref{wave}).
$u^\pm$ are eigenfunctions of  $J$ corresponding to the eigenvalues $\pm i$.
The projections $u \mapsto u^\pm$ correspond exactly to the decomposition of the
field into positive and negative energies.
\vskip 0.25cm
\noindent \textbf{II.} Assume from now on that $u$ is operator-valued. Our definitions and formula (\ref{main_formula})  imply
\begin{align}
& [u^+, u^+ ] = 0 = [u^-, u^-], \nonumber \\
& [ u^+(x), u^-(x^\prime) ] = -i G^+(x, x^\prime) \id, \nonumber \\
& [ u^-(x), u^+(x^\prime) ] = -i G^-(x, x^\prime) \id, \nonumber \\
& [ u(x), u(x^\prime) ] = -i G(x, x^\prime) \id,
\nonumber
\end{align}
where $\id$ is the identity operator, and the pointwise commutators
are interpreted as formally expressing the corresponding distributional
identity, e.g.
\begin{gather}
[ u(f), u(h) ] = -i \int_{\Om\times\Om} f(x) G(x,x^\prime) h(x^\prime) 
\sqrt{-|g|(x)} \sqrt{-|g|(x^\prime)} \, dx\, dx^\prime,
\nonumber
\end{gather}
where $f$ and $g$ are test functions.
\vskip 0.25cm
\noindent \textbf{III.} A theory of creation and annihilation operators can be constructed
from these formulas and (\ref{main_formula}). All the usual propagators can be 
deduced from the elementary kernels and $\cG$. 
\vskip 0.25cm
\noindent \textbf{IV.} In Minkowski space, the above decomposition agrees with 
the usual split of $u$ into negative and positive energy solutions via Fourier transform.
\vskip 0.25cm

One feature that stands out is the existence of positive 
energy solutions. This seems to be in contradiction with the
established fact that in a general spacetime there is no natural
notion of positive frequency (which is itself a consequence of the lack of
 uniqueness for the vacuum state in such situations). 
Recall, however, that such a notion 
is well defined for spacetimes that are (i) globally hyperbolic and (ii) stationary.
The former property has been implicitly employed by addressing the problem 
solely from a local perspective and considering 
$x, x^\prime \in D(\cS) \cap \Om$ in Assumption \ref{main_assump}. These 
 are not essential: when considering a global point of
view, it is rather natural to restrict oneself to  globally hyperbolic spacetimes,
and our lemmas and propositions can be extended to this  setting (recall
that we focused on local constructions because the approximation arguments
we used are local in nature).

Although not implicitly assumed, property (ii)  was ``almost" assumed, in the 
following sense. We have supposed the existence of the kernel $\cG$,  but sufficient  
conditions that  imply the validity of Assumption \ref{main_assump} have
not been given. It can be shown, however, that $\cG$ in fact exists for globally 
hyperbolic stationary spacetimes \cite{M}
(Euclidean at infinity for non-compact Cauchy surfaces;
see also Ref. \cite{C}). But, in the special case of globally hyperbolic 
stationary spacetimes, the 
decomposition into positive and negative energy solutions does in fact 
hold true \cite{W}.  Since it is not know, however, 
whether global hyperbolicity and stationarity are also necessary conditions for 
the existence of $\cG$ \cite{L3}, here we preferred to take the slightly more general 
point of view of postulating the existence of $\cG$ itself.

\end{document}